\documentclass[a4paper,numberwithinsect]{extarticle}

\usepackage{microtype}

\usepackage{amsmath,amsthm,amssymb}
\usepackage{graphicx}
\usepackage{hyperref}
\usepackage{fullpage}

\title{Reconfiguration of Colorable Sets in Classes of Perfect Graphs%
\footnote{Partially supported by MEXT/JSPS KAKENHI grant numbers
24106004, 
25730003. 
The second author was partially supported by FY 2015 Researcher Exchange Program between JSPS and NSERC.\@}}

\author{%
Takehiro Ito\thanks{%
Graduate School of Information Sciences, Tohoku University. Aoba-yama 6-6-05, Sendai, 980-8579, Japan.
\texttt{takehiro@ecei.tohoku.ac.jp}
}
\and Yota Otachi\thanks{%
Faculty of Advanced Science and Technology, Kumamoto University. 2-39-1 Kurokami, Chuo-ku, Kumamoto, 860-8555 Japan.
\texttt{otachi@cs.kumamoto-u.ac.jp}}}

\newcommand{\onestep}{\leftrightarrow}
\newcommand{\sevstep}{\leftrightsquigarrow}
\newcommand{\notsevstep}{\mathbin{{\leftrightsquigarrow}\hskip-.85em{/}\hskip.45em}}

\newcommand{\TAR}[1]{\mathsf{TAR}(#1)}
\newcommand{\TS}{\mathsf{TS}}
\newcommand{\TJ}{\mathsf{TJ}}

\newcommand{\TARrule}{\mathsf{TAR}}

\newcommand{\CSR}{\textsc{CSR}}

\newcommand{\symdiff}[2]{#1 \vartriangle #2}

\newcommand{\distTARG}[5]{\mathsf{dist}_{\TAR{#5}}(#3,#4)}
\newcommand{\distTJG}[4]{\mathsf{dist_{TJ}}(#3,#4)}
\newcommand{\distTSG}[4]{\mathsf{dist_{TS}}(#3,#4)}

\newtheorem{theorem}{Theorem}[section]
\newtheorem{lemma}[theorem]{Lemma}
\newtheorem{corollary}[theorem]{Corollary}

\theoremstyle{plain}
\newtheorem{proposition}[theorem]{Proposition}

\newtheorem{claim}[theorem]{Claim}

\newcommand{\figdir}{.} 
\newcommand{\figref}[1]{Figure~\ref{#1}}

\usepackage[table,xcdraw]{xcolor}

\usepackage{color}
\definecolor{lightblue}{rgb}{0.5,0.5,1.0}
\definecolor{darkred}{rgb}{0.8,0,0}
\definecolor{darkgreen}{rgb}{0,0.5,0}
\definecolor{darkblue}{rgb}{0,0,0.5}

\usepackage{algorithm}
\usepackage[noend]{algpseudocode}

\renewcommand{\mid}{:}

\begin{document}

\maketitle

\begin{abstract}
A set of vertices in a graph is \emph{$c$-colorable} if the subgraph induced by the set has a proper $c$-coloring.
In this paper, we study the problem of finding a step-by-step transformation (reconfiguration) between two $c$-colorable sets in the same graph.
This problem generalizes the well-studied \textsc{Independent Set Reconfiguration} problem.
As the first step toward a systematic understanding of the complexity of this general problem, we study the problem on classes of perfect graphs.
We first focus on interval graphs and give a combinatorial characterization of the distance between two $c$-colorable sets.
This gives a linear-time algorithm for finding an actual shortest reconfiguration sequence for interval graphs.
Since interval graphs are exactly the graphs that are simultaneously chordal and co-comparability,
we then complement the positive result by showing that even deciding reachability is PSPACE-complete
for chordal graphs and for co-comparability graphs.
The hardness for chordal graphs holds even for split graphs.
We also consider the case where $c$ is a fixed constant
and show that in such a case the reachability problem is polynomial-time solvable for split graphs
but still PSPACE-complete for co-comparability graphs.
The complexity of this case for chordal graphs remains unsettled.
As by-products, our positive results give the first polynomial-time solvable cases (split graphs and interval graphs)
for \textsc{Feedback Vertex Set Reconfiguration}.
\end{abstract}


\section{Introduction}
\label{sec:intro}

Recently, the reconfiguration framework has been applied to several search problems.
In a reconfiguration problem, we are given two feasible solutions of a search problem and are asked 
to determine whether we can modify one to the other by repeatedly applying prescribed reconfiguration rules while keeping the feasibility~(see \cite{ItoDHPSUU11,vandenHeuvel13,Nishimura17}).
Studying such a problem is important for understanding the structure of the solution space of the underlying problem.
Computational complexity of reconfiguration problems has been studied intensively.
For example, the \textsc{Independent Set Reconfiguration} problem under 
the reconfiguration rules $\TS$~\cite{HearnD05}, $\TARrule$~\cite{ItoDHPSUU11}, and $\TJ$~\cite{KaminskiMM12}
is studied for several graph classes such as
planar graphs~\cite{HearnD05},
perfect graphs~\cite{KaminskiMM12},
claw-free graphs~\cite{BonsmaKW14},
trees~\cite{DemaineDFHIOOUY15}, 
interval graphs~\cite{BonamyB17}, and
bipartite graphs~\cite{LokshtanovM18}.

In this paper, we initiate the study on problems of reconfiguring \emph{colorable sets}, which generalizes \textsc{Independent Set Reconfiguration}.
For a graph $G=(V,E)$ and an integer $c \ge 1$, a vertex set $S \subseteq V$ is \emph{$c$-colorable} if the subgraph $G[S]$ induced by $S$ admits a proper $c$-coloring.
For example, the $1$-colorable sets in a graph are exactly the independent sets of the graph.
Recently, $c$-colorable sets have been studied from the viewpoint of wireless network optimization
(see \cite{AsgeirssonHT17,BentertBN17arxiv} and the references therein).
The \textsc{Colorable Set Reconfiguration} problem asks
given two $c$-colorable sets $S$ and $S'$ in a graph $G$, 
whether we can reach from $S$ to $S'$ by repeatedly applying local changes allowed.
We consider the following three local change operations (see Section~\ref{sec:pre} for formal definitions): 
\begin{itemize}
  \item $\TAR{k}$: either adding or removing one vertex while keeping the size of the set at least a given threshold $k$.
  \item $\TJ$: swap one member for one nonmember.
  \item $\TS$: swap one member for one nonmember along an edge.
\end{itemize}

In perfect graphs, being $c$-colorable is equivalent to having no clique of size more than $c$.
This property often makes problems related to coloring tractable.
Thus, to understand this very general problem, we start the study on \textsc{Colorable Set Reconfiguration} with classes of perfect graphs.
\figref{fig:classes} shows the graph classes studied in this paper and the inclusion relationships
(see Section~\ref{subsec:graph-classes} for definitions).
\begin{figure}[thb]
  \centering
  \includegraphics[scale=.9]{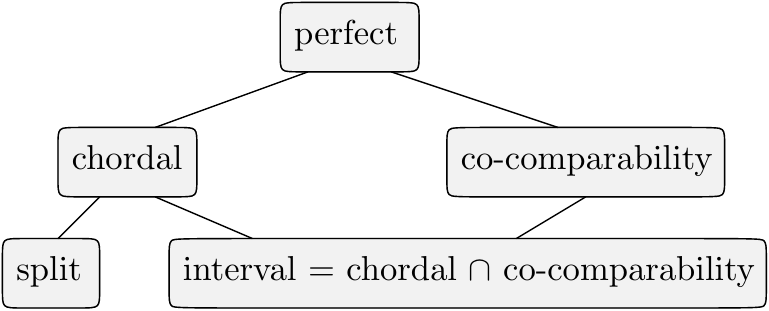}
  \caption{The graph classes studied in this paper.}
  \label{fig:classes}
\end{figure}

\subsection*{Our contribution}
Before we start our investigation on the reconfiguration problem,
we first fill a gap in the complexity landscape of the search problem \textsc{Colorable Set}
that asks for finding a large $c$-colorable set.
When $c=1$, \textsc{Colorable Set} is equivalent to the classical problem of finding a large independent set 
that can be solved in polynomial time for perfect graphs.
For larger $c$, it was only known that the case $c=2$ is NP-complete for perfect graphs~\cite{Addario-BerryKKLR10}.
To make the complexity status of \textsc{Colorable Set} for perfect graphs complete,
we show that it is NP-complete for any fixed $c \ge 2$ (Theorem~\ref{thm:cs_perfect}).

We then show complexity divergences among the classes of perfect graphs in \figref{fig:classes}, in particular under $\TARrule$ and $\TJ$.
See Table~\ref{tbl:summary} for a summary of our results.
Our results basically say that the problem under $\TARrule$ and $\TJ$ is tractable on interval graphs
but further generalization is not quite possible.

\begin{table}[]
\centering
\caption{Summary of the results. PSPACE-completeness results here apply to $\TS$ also, while polynomial-time algorithms do not.
The case of $c=1$ is equivalent to \textsc{Independent Set Reconfiguration.}}
\vspace*{.5ex}
\label{tbl:summary}
\begin{tabular}{c|c|c|c}
& \multicolumn{3}{c}{\cellcolor{lightgray!50}\textsc{Colorable Set Reconfiguration} under $\TARrule/\TJ$} \\
& \cellcolor{lightgray!25}$c=1$ & \cellcolor{lightgray!25}fixed $c \ge 2$ & \cellcolor{lightgray!25}arbitrary $c$ \\ \hline
\cellcolor{lightgray!50}perfect & \multicolumn{3}{c}{PSPACE-c} \\ \hline
\cellcolor{lightgray!50}co-comparability & PSPACE-c \cite{KaminskiMM12}\footnotemark & \multicolumn{2}{c}{PSPACE-c (Thm~\ref{thm:cocomp-hard})} \\ \hline
\cellcolor{lightgray!50}chordal & P \cite{KaminskiMM12} & \cellcolor{pink!60}? & PSPACE-c \\ \hline
\cellcolor{lightgray!50}split & P & P (Thm~\ref{thm:split_fixed}) & PSPACE-c (Thm~\ref{thm:split-hard}) \\ \hline
\cellcolor{lightgray!50}interval & P & \multicolumn{2}{c}{P (Thm~\ref{thm:interval})} \\ \hline
\cellcolor{lightgray!50}bipartite & NP-c \cite{LokshtanovM18} & \multicolumn{2}{c}{Trivial if $c \ge 2$} \\ \hline
\end{tabular}
\end{table}
\footnotetext{The reduction in \cite{KaminskiMM12} outputs co-comparability graphs. See also Theorem~\ref{thm:cocomp-hard} in this paper.}

More specifically, we first study the problem on interval graphs 
and show that a shortest reconfiguration sequence under $\TARrule$ can be found in linear time (Theorem~\ref{thm:interval}).
This implies the same result under $\TJ$.
Next we study the problem on split graphs. 
We show that the complexity depends on $c$.
When $c$ is a fixed constant, the problem is polynomial-time solvable under $\TARrule$ and $\TJ$ (Theorem~\ref{thm:split_fixed}).
If $c$ is a part of input, then we can show that the problem is PSPACE-complete under all rules, including $\TS$ (Theorem~\ref{thm:split-hard}).
While the hardness result applies also to chordal graphs,
it is unclear whether a similar positive result for chordal graphs can be obtained when $c$ is a fixed constant.
We only know that the case of $c=1$ under $\TARrule$ and $\TJ$ is polynomial-time solvable as chordal graphs are even-hole-free~\cite{KaminskiMM12}.
We finally show that for every fixed $c \ge 1$
the problem is PSPACE-complete for co-comparability graphs under all rules (Theorem~\ref{thm:cocomp-hard}).
Thus, our results are in some sense tight since the interval graphs are exactly the chordal co-comparability graphs
and split graphs are chordal graphs (see \figref{fig:classes}).


As a byproduct of Theorems~\ref{thm:interval} and \ref{thm:split_fixed}, 
the \textsc{Feedback Vertex Set Reconfiguration} problem~\cite{MouawadNRSS17} turns out to be polynomial-time solvable
for split graphs and interval graphs under $\TARrule$ and $\TJ$.
These are the first polynomial-time solvable cases for \textsc{Feedback Vertex Set Reconfiguration}.
To see the polynomial-time solvability, observe that the complements $V(G) \setminus S$ of $2$-colorable sets $S$ in a chordal graph $G$ are exactly the feedback vertex sets in the graph\footnote{%
Each induced cycle in a chordal graph is a triangle, and thus 2-colorable (or equivalently, odd cycle free) chordal graphs are forests.}
and reconfigurations of the complements are equivalent to reconfigurations of the original vertex sets under $\TARrule$ and $\TJ$.


\section{Preliminaries}
\label{sec:pre}

We say, as usual, that an algorithm for a graph $G = (V,E)$ runs in \emph{linear time}
if the running time of the algorithm is $O(|V| + |E|)$.

A \emph{proper $c$-coloring} of a graph assigns a color from $\{1,\dots,c\}$ to each vertex
in such a way that adjacent vertices have different colors.
Given a graph $G$ and an integer $c$, \textsc{Graph Coloring} asks whether $G$ admits a proper $c$-coloring.
This problem is NP-complete even if $c$ is fixed to 3~\cite{GareyJS76}.
The minimum $c$ such that a graph admits a proper $c$-coloring is its \emph{chromatic number}.

The \textsc{Colorable Set} problem is a generalization of \textsc{Graph Coloring}
where we find a large induced subgraph of the input graph that admits a proper $c$-coloring.
Let $G = (V,E)$ be a graph.
For a set of vertices $S \subseteq V$, we denote by $G[S]$ the subgraph induced by $S$.
A vertex set $S \subseteq V$ is \emph{$c$-colorable} in $G$ if $G[S]$ has a proper $c$-coloring.
Now the problem is defined as follows:
\begin{itemize}
  \setlength{\itemsep}{0pt}
  \item[] \textbf{Problem:} \textsc{Colorable Set}
  \item[] \textbf{Input:} A graph $G$ and integers $c$ and $k$.
  \item[] \textbf{Question:} Does $G$ have a $c$-colorable set of size at least $k$?
\end{itemize}

The problem of finding a large $c$-colorable set
is studied for a few important classes of perfect graphs (see \figref{fig:classes} and Table~\ref{tbl:summary}).
For the class of perfect graphs,
it is known that a maximum $1$-colorable set (that is, a maximum independent set) can be found in polynomial time~\cite{GrotschelLS88}.
Parameterized complexity~\cite{KrithikaN13} and approximation~\cite{FioriniKNR14} 
of \textsc{Colorable Set} on perfect graphs are also studied.

\subsection{Reconfiguration of colorable sets}

Let $S$ and $S'$ be $c$-colorable sets in a graph $G$.
Then, \emph{$S \onestep S'$ under $\TAR{k}$} for a nonnegative integer $k$
if $|S|, |S'| \ge k$ and $|\symdiff{S}{S'}| =  1$, where
$\symdiff{S}{S'}$ denotes the symmetric difference $(S \setminus S') \cup (S' \setminus S)$.
Here $S \onestep S'$ means that $S$ and $S'$ can be reconfigured to each other in one step
and $\TARrule$ stands for ``token addition \& removal.''
A sequence $\langle S_{0}, S_{1}, \dots, S_{\ell} \rangle$ of
$c$-colorable sets in $G$ is a \emph{reconfiguration sequence} of length $\ell$ between $S_{0}$ and $S_{\ell}$ under $\TAR{k}$
if $S_{i-1} \onestep S_{i}$ holds under $\TAR{k}$ for all $i \in \{1,2, \dots, \ell\}$.
A reconfiguration sequence under $\TAR{k}$ is simply called a \emph{$\TAR{k}$-sequence}.
We write \emph{$S_{0} \sevstep S_{\ell}$ under $\TAR{k}$} if there exists a $\TAR{k}$-sequence between $S_{0}$ and $S_{\ell}$.
Note that every reconfiguration sequence is \emph{reversible}, that is,
$S_{0} \sevstep S_{\ell}$ if and only if $S_{\ell} \sevstep S_{0}$.
Now the problem we are going to consider is formalized as follows:
\begin{itemize}
  \setlength{\itemsep}{0pt}
  \item[] \textbf{Problem:} \textsc{Colorable Set Reconfiguration} under $\TARrule$ ($\CSR_{\TARrule}$ for short)
  \item[] \textbf{Input:} A graph $G$, integers $c$ and $k$, and $c$-colorable sets $S$ and $S'$ of $G$.
  \item[] \textbf{Question:} Does $S \sevstep S'$ under $\TAR{k}$ hold?
\end{itemize}
We denote by $(G, c, S, S', k)$ an instance of $\CSR_{\TARrule}$.
We assume that both $|S| \ge k$ and $|S'| \ge k$ hold; otherwise it is trivially a no-instance.
Note that the lower bound $k$ guarantees that none of the sets in the reconfiguration sequence is too small.
Without the lower bound, the reachability problem becomes trivial as $S$ can always reach $S'$ via $\emptyset$.

For a $\CSR_\TARrule$-instance $(G, c, S, S', k)$,
we denote by $\distTARG{G}{c}{S}{S'}{k}$ the length of a shortest $\TAR{k}$-sequence
in $G$ between $S$ and $S'$; if there is no such a sequence, then
we set $\distTARG{G}{c}{S}{S'}{k} = \infty$.

We note that $\CSR_{\TARrule}$ is a decision problem and hence does not require the specification of an actual $\TAR{k}$-sequence.
Similarly, the shortest variant of $\CSR_{\TARrule}$ simply requires to output the value of $\distTARG{G}{c}{S}{S'}{k}$.

\subsubsection*{Other reconfiguration rules}
Although the $\TARrule$ rule is our main target,
we also study two other well-known rules $\TJ$ (token jumping) and $\TS$ (token sliding).
Let $S$ and $S'$ be $c$-colorable sets in a graph $G$.
For $\TJ$ and $\TS$, we additionally assume that $|S| = |S'|$
because these rules do not change the size of a set.
Now the rules are defined as follows:
\begin{itemize}
  \item $S \onestep S'$ under $\TJ$ if $|S \setminus S'| = |S' \setminus S| = 1$;
  
  \item $S \onestep S'$ under $\TS$ if $|S \setminus S'| = |S' \setminus S| = 1$ and
  the two vertices in $\symdiff{S}{S'}$ are adjacent in $G$.
\end{itemize}

Reconfiguration sequences under $\TJ$ and $\TS$
as well as the reconfiguration problems $\CSR_{\TJ}$ and $\CSR_{\TS}$ are defined analogously.
An instance of $\CSR_{\TJ}$ or $\CSR_{\TS}$ is represented as $(G,c,S,S')$,
and $\distTJG{G}{c}{S}{S'}$ and $\distTSG{G}{c}{S}{S'}$ are defined in the same way.

The following relation can be shown in almost the same way as Theorem~1 in~\cite{KaminskiMM12}
and means that $\CSR_{\TJ}$ is not harder than $\CSR_{\TARrule}$ in the sense of Karp reductions.
\begin{lemma}
\label{lem:TJ=TAR}
Let $S$ and $S'$ be $c$-colorable sets of size $k+1$ in a graph $G$.
Then, $S \sevstep S'$ under $\TAR{k}$ if and only if $S \sevstep S'$ under $\TJ$.
Furthermore, $\distTARG{G}{c}{S}{S'}{k} = 2 \cdot \distTJG{G}{c}{S}{S'}$ holds.
\end{lemma}

To make the presentation easier, we often use the shorthands $S+v$ for $S \cup \{v\}$ and $S-v$ for $S \setminus \{v\}$.
For a vertex $v$ of a graph $G$, we denote the neighborhood of $v$ in $G$ by $N_{G}(v)$.

\subsection{Graph classes}
\label{subsec:graph-classes}
A \emph{clique} in a graph is a set of pairwise adjacent vertices.
A graph is \emph{perfect} if the chromatic number equals the maximum clique size for every induced subgraph~\cite{Golumbic04}.
The following fact follows directly from the definition of perfect graphs
and will be used throughout this paper.
\begin{proposition}
\label{prop:colorable=noclique}
A vertex set $S \subseteq V(G)$ of a perfect graph $G$
is $c$-colorable if and only if $G[S]$ has no clique of size more than $c$.
\end{proposition}

There are many graph classes of perfect graphs.
Chordal graphs form one of the most well-known subclasses of perfect graphs,
where a graph is \emph{chordal} if it contains no induced cycle of length greater than 3.

Co-comparability graphs form another large class of perfect graphs.
A graph $G = (V,E)$ is a \emph{co-comparability graph} if there is a linear ordering $\prec$ on $V$
such that $u \prec v \prec w$ and $\{u,w\} \in E$ imply $\{u,v\} \in E$ or $\{v,w\} \in E$.
Although they are less known than chordal graphs,
co-comparability graphs generalize several important graph classes
such as interval graphs, permutation graphs, trapezoid graphs, and co-bipartite graphs (see \cite{Golumbic04,Spinrad03}).

The classes of chordal graphs and co-comparability graphs are incomparable.%
\footnote{A cycle of four vertices is a co-comparability graph but not chordal.
The \emph{net} graph obtained by attaching a pendant vertex to each vertex of a triangle is chordal but not a co-comparability graph.}
It is known that the class of interval graphs characterizes their intersection;
namely, a graph is an interval graph if and only if it is a co-comparability graph and chordal~\cite{GilmoreH64}.
Recall that a graph is an \emph{interval graph} if it is the intersection graph of closed intervals on the real line.

Another well-studied subclass of chordal graphs (and hence of perfect graphs) is the class of split graphs.
A graph $G = (V,E)$ is a \emph{split graph} if $V$ can be partitioned into a clique $K$ and an independent set $I$.
To emphasize that $G$ is a split graph, we write $G = (K,I; E)$.
The classes of interval graphs and split graphs are incomparable.%
\footnote{A path with five or more vertices is an interval graph but not a split graph.
The net graph is a split graph but not an interval graph.}


\section{NP-hardness of \textsc{Colorable Set} on perfect graphs for fixed $c \ge 2$}

It is known that 
if $c$ is unbounded, \textsc{Colorable Set} is polynomial-time solvable
for interval graphs~\cite{YannakakisG87,MaratheRR92} and
more generally for co-comparability graphs~\cite{Frank80},
while it is NP-complete for split graphs (and thus for chordal graphs)~\cite{YannakakisG87,CorneilF89}.
On the other hand, if $c$ is a fixed constant,
\textsc{Colorable Set} is polynomial-time solvable even for chordal graphs~\cite{YannakakisG87,CorneilF89}.

For perfect graphs,  the case of $c=1$ is solvable in polynomial time~\cite{GrotschelLS88},
while the case of $c=2$ is NP-complete~\cite{Addario-BerryKKLR10}.
Here we demonstrate that the problem is hard for any fixed $c$.
\begin{theorem}
\label{thm:cs_perfect}
\textsc{Colorable Set} is NP-complete on perfect graphs for every fixed $c \ge 2$.
\end{theorem}

In \cite{Addario-BerryKKLR10}, the problem actually studied was the dual of our problem.
An \emph{odd cycle} is a cycle of odd length.
An \emph{odd cycle transversal} $S \subseteq V$ of a graph $G = (V,E)$ is a set of vertices
that intersects every cycle of odd length in $G$. 
In other words, $S$ is an odd cycle transversal if and only if $G[V \setminus S]$ is bipartite.
They study the following problem of finding a small odd cycle transversal:
\begin{itemize}
  \setlength{\itemsep}{0pt}
  \item[] \textbf{Problem:}  \textsc{Odd Cycle Transversal} (OCT)
  \item[] \textbf{Input:} A graph $G$ and an integer $k$.
  \item[] \textbf{Question:} Does $G$ have an odd cycle transversal of size at most $k$?
\end{itemize}
\begin{proposition}
[\cite{Addario-BerryKKLR10}]
OCT is NP-complete for perfect graphs.
\end{proposition}

The \emph{join} of two disjoint graphs $G = (V_{G}, E_{G})$ and $H = (V_{H}, E_{H})$
is the graph $G \oplus H = (V_{G} \cup V_{H}, E_{G} \cup E_{H} \cup \{\{u, v\} \mid u \in V_{G}, v \in V_{H}\})$.
That is, $G \oplus H$ is obtained from the disjoint union of $G$ and $H$ by adding all possible edges between $V_{G}$ and $V_{H}$.
\begin{lemma}
\label{lem:perfect_join}
The class of perfect graphs is closed under join.
That is, if two disjoint graphs are perfect, then so is their join.
\end{lemma}
\begin{proof}
An \emph{odd hole} is an induced odd cycle of length at least 5.
Let $G$ and $H$ be disjoint perfect graphs.
By the strong perfect graph theorem~\cite{ChudnovskyRST06},
it suffices to show that none of $G \oplus H$ and 
its complement $\overline{G \oplus H}$ contains an odd hole.

Suppose to the contrary that $G \oplus H$ contains an odd hole $C$.
Since $G$ and $H$ are perfect, $C$ intersects both $G$ and $H$. 
Moreover, as $|V(C)| \ge 5$, one of $G$ and $H$ has at least three vertices in $C$.
Therefore, $(G \oplus H)[V(C)]$ contains a vertex of degree at least 3.
This contradicts that $C$ is an induced cycle and $(G \oplus H)[V(C)] = C$.

Observe that $\overline{G \oplus H}$ is the disjoint union of 
the complements $\overline{G}$ of $G$ and $\overline{H}$ of $H$.
By the (weak) perfect graph theorem~\cite{Lovasz72},
$\overline{G}$ and $\overline{H}$ are perfect and thus have no odd hole.
Hence, $\overline{G \oplus H}$ has no odd hole.
\end{proof}

Now we are ready for proving the main claim of this section.
\begin{proof}
[Proof of Theorem~\ref{thm:cs_perfect}]
Let $G$ be an instance of OCT for perfect graphs.
Let $H$ be the disjoint union of $n := |V(G)|$ cliques of size $c-2$.
Let $G' = G \oplus H$. 
By Lemma~\ref{lem:perfect_join}, $G'$ is perfect (as $H$ is clearly perfect).
Now it suffices to show that for any $k < n$, $G$ has an odd cycle transversal of size $k$
if and only if $G'$ has a $c$-colorable set of size $|V(G')| - k$.
 
To show the only-if part, let $S$ be an odd cycle transversal of size $k$ in $G$.
Let $T = V(G') \setminus S$. Clearly, $|T| = |V(G')| - k$.
Since $G[V(G) \setminus S]$ contains no clique of size 3, the maximum clique size of $G'[T]$ is at most $c$.
Since $G'$ is perfect, $T$ is $c$-colorable.

To show the if part, let $T$ be a $c$-colorable set of size $|V(G')| - k$ in $G'$.
Observe that $T$ contains at least one entire clique $K$ of size $c-2$ in $H$ since otherwise $|T| \le |V(G')| - n < |V(G')| - k$.
Let $S = V(G') \setminus T$. 
Suppose to the contrary that $G[V(G) \setminus S]$ is not bipartite and thus contains a clique $K'$ of size 3.
By the definition, $K' \subseteq T \cap V(G)$.
Thus, $T$ contains the clique $K \cup K'$ of size $c+1$. This contradicts the $c$-colorability of $T$.
\end{proof}


\section{Shortest reconfiguration in interval graphs}
\label{sec:interval}
In this section, we show that $\CSR_{\TARrule}$ for interval graphs can be solved in linear time.
Our result is actually stronger and says that an actual shortest $\TAR{k}$-sequence can be found in linear time, if any.
By Lemma~\ref{lem:TJ=TAR}, the same result is obtained for $\TJ$.
We first give a characterization of the distance between two $c$-colorable sets in an interval graph (Section~\ref{ssec:int-dist}).
This characterization says that a shortest $\TAR{k}$-sequence has length linear in the number of vertices of the graph.
We then show that the distance can be computed in linear time (Section~\ref{ssec:int-dist-comp}).
We finally present a linear-time algorithm for finding a shortest $\TAR{k}$-sequence (Section~\ref{ssec:int-seq}).

It is known that a graph is an interval graph if and only if its maximal cliques can be ordered so that
each vertex appears consecutively in that ordering~\cite{GilmoreH64,FulkersonG65}.
We call a list of the maximal cliques ordered in such a way a \emph{clique path}.
Let $G = (V,E)$ be an interval graph and $(M_{1}, \dots, M_{t})$ be a clique path of $G$;
that is, for each vertex $v \in V$, there are indices $l_{v}$ and $r_{v}$ such that $v \in M_{i}$
if and only if $l_{v} \le i \le r_{v}$. 
Given an interval graph, a clique path and the indices $l_{v}$ and $r_{v}$ for all vertices can be computed in linear time~\cite{UeharaU07}.
Hence we can assume that we are additionally given such information.
Note that $\mathcal{I} = \{[l_{v}, r_{v}] \mid v \in V\}$ is an interval representation of $G$.
Namely, $\{u,v\} \in E$ if and only $[l_{u}, r_{u}] \cap [l_{v}, r_{v}] \ne \emptyset$.

Let $K$ be a clique in an interval graph $G$.
By the Helly property of intervals, the intersection of all intervals in $K$ is nonempty;
that is,  $\bigcap_{v \in K} [l_{v}, r_{v}] \ne \emptyset$ (see \cite{Spinrad03}).
A point in the intersection $\bigcap_{v \in K} [l_{v}, r_{v}]$ is a \emph{clique point} of $K$.


\subsection{The distance between $c$-colorable sets}
\label{ssec:int-dist}

Let $(G, c, S, S', k)$ be an instance of $\CSR_\TARrule$.
The set $S$ is \emph{locked in $G$}
if $S$ is a maximal $c$-colorable set in $G$ and $|S| = k$.
The following lemma follows immediately from the definition.
\begin{lemma}
\label{lem:interval-unreachable} 
Let $G$ be a graph, 
and let $S$ and $S'$ be distinct $c$-colorable sets of size at least $k$ in $G$.
If $S$ or $S'$ is locked in $G$, then $S \notsevstep S'$.
\end{lemma}
\begin{proof}
Assume without loss of generality that $S$ is locked in $G$.
If there is a $c$-colorable set $S_{1}$ in $G$ such that $S \onestep S_{1}$,
then $S \subsetneq S_{1}$ as $|S| = k$. This contradicts the maximality of $S$.
Since $S \ne S'$, we can conclude that $S \notsevstep S'$.
\end{proof}

The rest of this subsection is dedicated to a proof of the following theorem,
which implies that the converse of the lemma above also holds for interval graphs.

\begin{theorem}
\label{thm:dist-interval}
Let $G$ be an interval graph, 
and let $S$ and $S'$ be distinct $c$-colorable sets of size at least $k$ in $G$.
If $S$ and $S'$ are not locked in $G$,
then the distance $d := \distTARG{G}{c}{S}{S'}{k}$ is determined as follows.
\begin{enumerate}
  \item If $S$ and $S'$ are not locked in $G[S \cup S']$, then $d = |\symdiff{S}{S'}|$. 
  \label{itm:interval_nolocked}
  
  \item If exactly one of $S$ and $S'$ is locked in $G[S \cup S']$, then $d = |\symdiff{S}{S'}| + 2$.
  \label{itm:interval_onelocked}

  \item If $S$ and $S'$ are locked in $G[S \cup S']$, then we have the following two cases.
  \label{itm:interval_twolocked}
  \begin{enumerate}
    \item If there is $v \in V(G) \setminus (S \cup S')$ such that both $S + v$ and $S' + v$ are $c$-colorable in $G$,
    then $d = |\symdiff{S}{S'}| + 2$.
    \label{itm:interval_twolocked_common}

    \item Otherwise, 
    $d = |\symdiff{S}{S'}| + 4$.
    \label{itm:interval_twolocked_no_common}
  \end{enumerate}
 \end{enumerate}
\end{theorem}

\begin{corollary}
For $S \ne S'$, $S \sevstep S'$ if and only if none of $S$ and $S'$ is locked in $G$.
\end{corollary}

Observe that $\distTARG{G}{c}{S}{S'}{k} \ge |\symdiff{S}{S'}|$ for any pair of $c$-colorable sets $S$ and $S'$ in $G$.
We use this fact implicitly in the following arguments.

\begin{lemma}[Theorem~\ref{thm:dist-interval}~(\ref{itm:interval_nolocked})]
\label{lem:interval_nolocked}
Let $G$ be an interval graph, 
and let $S$ and $S'$ be $c$-colorable sets of size at least $k$ in $G$.
If $S$ and $S'$ are not locked in $G[S \cup S']$,
then $\distTARG{G}{c}{S}{S'}{k} = |\symdiff{S}{S'}|$.
\end{lemma}
\begin{proof}
We proceed by induction on $|\symdiff{S}{S'}|$.
The base case of $|\symdiff{S}{S'}| = 0$ is trivial.
Assume that $|\symdiff{S}{S'}| > 0$
and that the statement is true if the symmetric difference is  smaller.

We first consider the case where $|S| = k$.
Since $S$ is not locked in $G[S \cup S']$, $S$ is not maximal in $G[S \cup S']$.
Thus there is a vertex $v \in S' \setminus S$ such that $T := S + v$ is $c$-colorable.
The set $T$ is not locked in $G[S \cup S']$, $S \onestep T$, and $|\symdiff{T}{S'}| = |\symdiff{S}{S'}| - 1$.
By the induction hypothesis, $\distTARG{G}{c}{T}{S'}{k} = |\symdiff{T}{S'}| = |\symdiff{S}{S'}| - 1$.
Hence, we have $\distTARG{G}{c}{S}{S'}{k} \le \distTARG{G}{c}{T}{S'}{k} + 1 = |\symdiff{S}{S'}|$.
If $|S'| = k$, we can apply the same argument.

In the following, we assume that $|S| > k$ and $|S'| > k$.
If $S \subseteq S'$, then 
we can add the elements of $S' \setminus S$ one-by-one in an arbitrary order
to get a shortest reconfiguration sequence of length $|S' \setminus S| = |\symdiff{S}{S'}|$.
The case where $S' \subseteq S$ is the same.

We now consider the case where $S \not\subseteq S'$ and $S' \not\subseteq S$.
Let $v \in S \setminus S'$ and $w \in S' \setminus S$ be vertices with the smallest right-end in each set.
That is, $r_{v} = \min \{r_{x} \mid x \in S \setminus S'\}$
and $r_{w} = \min \{r_{x} \mid x \in S' \setminus S\}$.
By symmetry, assume that $r_{w} \le r_{v}$.
Let $u \in S \setminus S'$ be a vertex that minimizes $l_{u}$.
(Note that $u$ and $v$ may be the same.)
Now we have $r_{w} \le r_{v} \le r_{u}$.
We set $T = S - u$ and $T' = S -u + w$. Clearly, $S \onestep T$.
To apply induction hypothesis, it suffices to show that $T$ is not locked in $G[T \cup S']$.
To this end, we prove that $T \onestep T'$.
Suppose to the contrary that $T'$ is not $c$-colorable;
that is, $T'$ contains a clique $K$ of size $c+1$.
Since $T$ does not contain such a large clique, $K$ must include $w$.
Let $p$ be a clique point of $K$.
If $p < l_{u}$, then $K$ includes no vertex in $S \setminus S'$
as $u$ has the minimum $l_{u}$ in $S \setminus S'$.
This contradicts the $c$-colorability of $S'$
and thus $l_{u} \le p \le r_{w} \le r_{u}$.
This implies that $K - w + u \subseteq S$ is a clique of size $c+1$, a contradiction.
Therefore, we can conclude that $T'$ is $c$-colorable.
Now, by the induction hypothesis, $\distTARG{G}{c}{T}{S'}{k} = |\symdiff{T}{S'}| = |\symdiff{S}{S'}| - 1$,
and thus $\distTARG{G}{c}{S}{S'}{k} \le \distTARG{G}{c}{T}{S'}{k} + 1 = |\symdiff{S}{S'}|$.
\end{proof}

\begin{lemma}
[Theorem~\ref{thm:dist-interval}~(\ref{itm:interval_onelocked})]
\label{lem:interval_onelocked}
Let $G$ be an interval graph, 
and let $S$ and $S'$ be distinct $c$-colorable sets of size at least $k$ in $G$.
If $S$ and $S'$ are not locked in $G$,
and exactly one of $S$ and $S'$ is locked in $G[S \cup S']$, then $\distTARG{G}{c}{S}{S'}{k} = |\symdiff{S}{S'}| + 2$. 
\end{lemma}
\begin{proof}
Without loss of generality, assume that $S$ is locked in $G[S \cup S']$.
This implies that $|S| = k$.
Since $S$ is not locked in $G$, $S$ is not maximal in $G$.
Hence, there is a vertex $v \in V(G) \setminus (S \cup S')$
such that $T := S + v$ is a $c$-colorable set of $G$.
Observe that $T$ and $S'$ are not locked in $G[T \cup S']$.
Thus, by Theorem~\ref{thm:dist-interval}~(\ref{itm:interval_nolocked}), it holds that
$\distTARG{G}{c}{S}{S'}{k} \le \distTARG{G}{c}{T}{S'}{k} + 1 = |\symdiff{T}{S'}| + 1 = |\symdiff{S}{S'}| + 2$.

On the other hand, 
since $S$ is locked in $G[S \cup S']$,
every $c$-colorable set $T$ of $G$ with $S \onestep T$
contains a vertex in $V(G) \setminus (S \cup S')$.
Thus $|\symdiff{T}{S'}| = |\symdiff{S}{S'}|+1$ holds.
This implies that
$\distTARG{G}{c}{S}{S'}{k} \ge \min_{T \colon S \onestep T} |\symdiff{T}{S'}| + 1 = |\symdiff{S}{S'}|+2$.
\end{proof}

\begin{lemma}
[Theorem~\ref{thm:dist-interval}~(\ref{itm:interval_twolocked_common})]
\label{lem:interval_twolocked_common}
Let $G$ be an interval graph, 
and let $S$ and $S'$ be distinct $c$-colorable sets of size at least $k$ in $G$.
Assume $S$ and $S'$ are locked in $G[S \cup S']$ but not in $G$.
If there is a vertex $v \in V(G) \setminus (S \cup S')$ such that both $S + v$ and $S' + v$ are $c$-colorable in $G$,
then $\distTARG{G}{c}{S}{S'}{k} = |\symdiff{S}{S'}| + 2$.
\end{lemma}
\begin{proof}
Let $v \in V(G) \setminus (S \cup S')$ be a vertex such that both $S + v$ and $S' + v$ are $c$-colorable in $G$.
We have $S \onestep S + v$ and $S' \onestep S' + v$.
Since $S + v$ and $S' + v$ are not locked in $G[S \cup S' + v]$,
Theorem~\ref{thm:dist-interval}~(\ref{itm:interval_nolocked}) implies that
$\distTARG{G}{c}{S}{S'}{k} \le \distTARG{G}{c}{S+v}{S'+v}{k} + 2 = |\symdiff{(S+v)}{(S'+v)}| + 2 = |\symdiff{S}{S'}| + 2$.

The lower bound $\distTARG{G}{c}{S}{S'}{k} \ge |\symdiff{S}{S'}|+2$
can be shown in exactly the same way as in the proof of Lemma~\ref{lem:interval_onelocked}.
\end{proof}

\begin{lemma}
[Theorem~\ref{thm:dist-interval}~(\ref{itm:interval_twolocked_no_common})]
\label{lem:interval_twolocked_no_common}
Let $G$ be an interval graph, 
and let $S$ and $S'$ be distinct $c$-colorable sets of size at least $k$ in $G$.
Assume $S$ and $S'$ are locked in $G[S \cup S']$ but not in $G$.
If there is no vertex $v \in V(G) \setminus (S \cup S')$ such that both $S + v$ and $S' + v$ are $c$-colorable in $G$,
then $\distTARG{G}{c}{S}{S'}{k} = |\symdiff{S}{S'}| + 4$.
\end{lemma}
\begin{proof}
Let $u, v \in V(G) \setminus (S \cup S')$ be distinct vertices such that $S + u$ and $S' + v$ are $c$-colorable in $G$.
Since $S + u$ and $S' + v$ are not locked in $G[(S + u) \cup (S' + v)]$,
Theorem~\ref{thm:dist-interval}~(\ref{itm:interval_nolocked}) implies that
$\distTARG{G}{c}{S}{S'}{k} \le \distTARG{G}{c}{S+u}{S'+v}{k} + 2 = |\symdiff{(S+u)}{(S'+v)}| + 2 = |\symdiff{S}{S'}| + 4$.

Since $S$ and $S'$ are locked in $G[S \cup S']$ and 
there is no vertex $v \in V(G) \setminus (S \cup S')$ such that both $S + v$ and $S' + v$ are $c$-colorable in $G$,
we have $\min_{T \colon S \onestep T, \ T' \colon S' \onestep T'} |\symdiff{T}{T'}| = |\symdiff{S}{S'}|+2$.
This implies that $\distTARG{G}{c}{S}{S'}{k} \ge |\symdiff{S}{S'}|+4$.
\end{proof}


\subsection{Computing the distance in linear time}
\label{ssec:int-dist-comp}

We here explain how to check which case of Theorem~\ref{thm:dist-interval} applies to a given instance in linear time.
\begin{lemma}
\label{lem:inerval_maximality_of_two}
Given an interval graph $G$ and $c$-colorable sets $S$ and $S'$ in $G$,
one can either find a vertex $v \notin S \cup S'$ such that $S+v$ and $S'+v$ are $c$-colorable
or decide that no such vertex exists in linear time.
\end{lemma}
\begin{proof}
Let $(M_{1}, \dots, M_{t})$ be a clique path of $G$.
Recall that $M_{1}, \dots, M_{t}$ are the maximal cliques of $G$.
Thus, for every $T \subseteq V(G)$, the maximum clique size of $G[T]$
is equal to $\max_{1 \le i \le t} |T \cap M_{i}|$.

We compute $a_{i}^{S} = |S \cap M_{i}|$ for $1 \le i \le t$ as follows.
Initialize all $a_{i}^{S}$ to $0$;
for each $u \in S$, add $1$ to all $a_{i}^{S}$ with $l_{u} \le i \le r_{u}$.
In the same way, we compute $a_{i}^{S'} = |S' \cap M_{i}|$ for $1 \le i \le t$.
From the observation above, we can conclude that 
for each vertex $v \notin S \cup S'$, $S+v$ and $S'+v$ are $c$-colorable
if and only if $a_{i}^{S}, a_{i}^{S'}  < c$ for $l_{v} \le i \le r_{v}$.

The initialization and the test for all nonmembers of $S$ can be done in time $O(\sum_{i=1}^{t} |M_{i}|)$.
It suffices to show that $\sum_{i=1}^{t} |M_{i}| \le \sum_{v \in V(G)}(\deg(v) + 1) = 2|E(G)| + |V(G)|$.
Since $M_{1} \not\subseteq M_{2}$, there is a vertex $v$ with $l_{v} = r_{v} = 1$.
Thus $|M_{1}| = \deg(v) + 1$. By induction on the number of vertices, our claim holds.
\end{proof}

By setting $S = S'$ in the lemma above, we have the following lemma.
\begin{lemma}
\label{lem:inerval_maximality}
Given an interval graph $G$ and a $c$-colorable set $S$ in $G$,
one can either find a vertex $v \notin S$ such that $S +v$ is $c$-colorable
or decide that $S$ is maximal in linear time.
\end{lemma}

\begin{corollary}
\label{cor:inerval_distance}
Given an interval graph $G$ and $c$-colorable sets $S$ and $S'$ in $G$,
the distance $\distTARG{G}{c}{S}{S'}{k}$ can be computed in linear time.
\end{corollary}
\begin{proof}
We first check whether $S$ or $S'$ is locked in $G$. If so, the distance is $\infty$.
Otherwise, we check whether $S$ and $S'$ are locked in $G[S \cup S']$.
If not both of them are locked in $G[S \cup S']$,
then we can apply Theorem~\ref{thm:dist-interval}~(\ref{itm:interval_nolocked}) or (\ref{itm:interval_onelocked})
and determine the distance.
If both $S$ and $S'$ are locked in $G[S \cup S']$,
we find a vertex $v \notin S \cup S'$ such that both $S + v$ and $S' + v$ are $c$-colorable in $G$.
Everything can be done in linear time by Lemmas~\ref{lem:inerval_maximality_of_two} and \ref{lem:inerval_maximality}.
\end{proof}

\subsection{Finding a shortest reconfiguration sequence in linear time}
\label{ssec:int-seq}

Here we describe how we find an actual shortest reconfiguration sequence in linear time.
To this end, we need to be careful about the representation of a reconfiguration sequence.
If we always output the whole set, the total running time cannot be smaller than $k \cdot \distTARG{G}{c}{S}{S'}{k}$.
However, this product can be quadratic.
To avoid this blow up, we output only the difference from the previous set.
That is, if the current set is $S$ and the next set is $S + v$ ($S - v$), we output $+ v$ ($-v$, resp.).
We also fully use the reversible property of reconfiguration sequences
and output them sometimes from left to right and sometimes from right to left.
For example, we may output a reconfiguration sequence $\langle S_{0}, \dots, S_{5} \rangle$
as first $S_{0} \onestep S_{1} \onestep S_{2}$, next $S_{5} \onestep S_{4} \onestep S_{3}$, then $S_{2} \onestep S_{3}$.
It is straightforward to output the sequence from left to right by using a linear-size buffer.

\begin{theorem}
\label{thm:interval}
Given an interval graph $G$ and $c$-colorable sets $S$ and $S'$ in $G$,
a $\TAR{k}$-sequence of length $\distTARG{G}{c}{S}{S'}{k}$ can be computed in linear time.
\end{theorem}
\begin{proof}
We first test which case of Theorem~\ref{thm:dist-interval} applies to the given instance.
This can be done in linear time as shown in the proof of Corollary~\ref{cor:inerval_distance}.
We reduce Cases~(\ref{itm:interval_onelocked}) and (\ref{itm:interval_twolocked})
to Case~(\ref{itm:interval_nolocked}).
The reductions below can be done in linear time
by using Lemmas~\ref{lem:inerval_maximality_of_two} and \ref{lem:inerval_maximality}.

Assume first that Case~(\ref{itm:interval_onelocked}) applies;
that is, $S$ is locked but $S'$ is not in $G[S \cup S']$.
We find a vertex $v \in V(G) \setminus (S \cup S')$
such that $S + v$ is a $c$-colorable set of $G$.
We then add $v$ to $S$.
As we saw in the proof of Lemma~\ref{lem:interval_onelocked},
this is a valid step in a shortest reconfiguration sequence.
Furthermore, after this step, $S$ and $S'$ are not locked in $G[S \cup S']$.

Next assume that Case~(\ref{itm:interval_twolocked}) applies;
that is, both $S$ and $S'$ are locked in $G[S \cup S']$.
We find vertices $u, v \notin S \cup S'$ such that
$S + u$ and $S' + v$ are $c$-colorable in $G$.
In Case (\ref{itm:interval_twolocked_common}), we further ask that $u = v$.
We then add $u$ to $S$ and $v$ to $S'$.
The proofs of Lemmas~\ref{lem:interval_twolocked_common} and \ref{lem:interval_twolocked_no_common}
imply that these are valid steps in a shortest reconfiguration sequence,
and that $S$ and $S'$ are no longer locked in $G[S \cup S']$ after these steps.

We now handle Case (\ref{itm:interval_nolocked}),
where $S$ and $S'$ are not locked in $G[S \cup S']$.
Assume that $S \not\subseteq S'$ and $S' \not\subseteq S$
since otherwise finding a shortest sequence is trivial.
We first compute two orderings of the vertices in $S \cup S'$:
nondecreasing orderings of left-ends $l_{v}$ and of right-ends $r_{v}$.
Such orderings can be constructed in linear time from a clique path.
We maintain information for each vertex $v$ whether
$v \in S \setminus S'$, $v \in S' \setminus S$, or $v \notin \symdiff{S}{S'}$.
Using this information, we can also maintain vertices of the smallest left-end and of the smallest right-end
in each of $S \setminus S'$ and $S' \setminus S$.

Let $v \in S \setminus S'$ and $w \in S' \setminus S$ be vertices with the smallest right-end in each set.
By symmetry, assume that $r_{w} \le r_{v}$. Let $u \in S \setminus S'$ be a vertex that minimizes $l_{u}$.
As shown in the proof of Lemma~\ref{lem:interval_nolocked},
$S \onestep (S - u) \onestep (S -u + w)$ under $\TAR{k}$ and $|\symdiff{S}{S'}| = |\symdiff{(S -u + w)}{S'}| + 2$.
We output the two steps $S-u$ and $S -u + w$.

We then set $S:= S - u + w$ and update the information as $u, w \notin \symdiff{S}{S'}$ anymore.
We also have to maintain the vertices of the smallest left- and right-ends in each $S \setminus S'$ and $S' \setminus S$.
Let $w' \in S' \setminus S$ be a vertex with the smallest right-end.
The vertex $w$ can be found by sweeping the nondecreasing ordering of the right-ends from the position of $w$ to the right.
The vertex $u' \in S \setminus S'$ with the smallest left-end can be found in an analogous way.
Although a single update can take super constant steps, it sums up to a linear number of steps in total
since it can be seen as a single left-to-right scan of each nondecreasing ordering.
Therefore, the total running time is linear.
\end{proof}


\section{Split graphs}
For split graphs, we consider two cases.
In the first case, we assume that $c$ is a fixed constant, and show that the problem under $\TARrule$ (and $\TJ$) can be solved in $O(n^{c+1})$ time.
The second case is the general problem having $c$ as a part of input.
We show that in this case the problem is PSPACE-complete under all reconfiguration rules.

\subsection{Polynomial-time algorithm for fixed $c$}

Let $G = (K, I; E)$ be a split graph, where $K$ is a clique and $I$ is an independent set.
For $C \subseteq K$ with $|C| \le c$, we define $T_{C}$ as follows:
\[
  T_{C}
  = 
  \begin{cases}
    C \cup I & \text{if } |C| < c, \\
    C \cup I \setminus \{u \in I \mid C \subseteq N_{G}(u)\} & \text{if } |C| = c.
  \end{cases}
\]

We can see that $T_{C}$ is $c$-colorable for every $C \subseteq K$ with $|C| \le c$ as follows.
Every clique $K' \subseteq T_{C}$ includes at most $|C| \le c$ vertices in $C$ and at most one vertex in $I$.
Since a vertex in $T_{C} \cap I$ has fewer than $c$ neighbors in $C$, 
the maximum clique size of $G[T_{C}]$ is at most $c$.

\begin{lemma}
\label{lem:split_to-max}
If $S$ is a $c$-colorable set of $G$ with $|S| \ge k$,
then $S \sevstep T_{S \cap K}$ under $\TAR{k}$.
\end{lemma}
\begin{proof}
Note that $T_{S \cap K}$ is $c$-colorable since $S$ is $c$-colorable and thus $|S \cap K| \le c$.
We now show that $S \subseteq T_{S \cap K}$, which implies $S \sevstep T_{S \cap K}$.

If $|S \cap K| < c$, then $T_{S \cap K} = (S \cap K) \cup I$, and thus $S \subseteq T_{S \cap K}$.
If $|S \cap K| = c$, then $S \cap \{u \in I \mid (S \cap K) \subseteq N(u)\} = \emptyset$ since $S$ is $c$-colorable.
Thus it holds that $S \subseteq (S \cap K) \cup (I \setminus \{u \in I \mid (S \cap K) \subseteq N(u)\}) = T_{S \cap K}$.
\end{proof}

By the reversibility of reconfiguration sequences,
we can reduce the problem as follows.
\begin{corollary}
\label{cor:split_maximal}
If $S$ and $S'$ are $c$-colorable sets of $G$ with $|S| \ge k$ and $|S'| \ge k$,
then $S \sevstep S'$ under $\TAR{k}$ if and only if $T_{S \cap K} \sevstep T_{S' \cap K}$ under $\TAR{k}$.
\end{corollary}

Now we state the crucial lemma for solving the reduced problem.
\begin{lemma}
\label{lem:split-edge}
Let $C \subseteq K$ and $v \in K \setminus C$.
If $T_{C}$ and $T_{C+v}$ are $c$-colorable sets of size at least $k$,
then $T_{C} \sevstep T_{C+v}$ under $\TAR{k}$ if and only if $|T_{C+v}| \ge k+1$.
\end{lemma}
\begin{proof}
To prove the if part, assume that $|T_{C+v}| \ge k+1$.
Then, $T_{C+v} \onestep T_{C+v} -v$.
Since $(T_{C+v} -v) \cap K = C$, it holds that $T_{C+v} -v \sevstep T_{C}$ by Lemma~\ref{lem:split_to-max}.
Thus we have $T_{C+v} \sevstep T_{C}$.

To prove the only-if part, assume that $|T_{C+v}| = k$.
If $|C+v| = c$, then $T_{C+v}$ is a maximal $c$-colorable set
and no other $c$-colorable set of size at least $k$ can be reached from $T_{C+v}$.
Assume that $|C+v| < c$, and hence $|C| < c$.
Then, $T_{C+v} = (C+v) \cup I$ and $T_{C} = C \cup I$.
Therefore, we have $|T_{C}| = |T_{C+v} - v| = k-1$, a contradiction.
\end{proof}

Combining the arguments in this subsection, we are now ready to present 
a polynomial-time algorithm.
\begin{theorem}
\label{thm:split_fixed}
Given an $n$-vertex split graph $G = (K,I;E)$ and $c$-colorable sets $S$ and $S'$ of size at least $k$ in $G$,
it can be decided whether $S \sevstep S'$ under $\TAR{k}$ in time $O(n^{c+1})$.
\end{theorem}
\begin{proof}
We construct a graph $H = (\mathcal{K}, \mathcal{E})$ from $G = (K,I;E)$ as follows:
\begin{align*}
  \mathcal{K} &= \{C \subseteq K \mid |C| \le c \text{ and } |T_{C}| \ge k\}, \\
  \mathcal{E} &= \{\{C, C+v\} \mid C, C+v \in \mathcal{K} \text{ and } |T_{C+v}| \ge k + 1\}.
\end{align*}
For each $C \subseteq K$ with $|C| < c$, the size $|T_{C}| = |C| + |I|$ can be computed in constant time
(assuming that we know the size $|I|$ in advance).
If $|C| = c$, then we need to compute the size of $\{u \in I \mid C \subseteq N(u)\}$.
This can be done in time $O(n)$ for each $C$. 
In $H$, each $C \in \mathcal{K}$ is adjacent to at most $|C|$ subsets of $C$:
if $|T_{C}| > k$, then $C$ is adjacent to all $C - v$ with $v \in C$;
otherwise, $C$ has no edge to its subsets. This can be computed in time $O(n)$ for each $C$.
In total, the graph $H$ with $O(n^{c})$ vertices and $O(n^{c})$ edges can be constructed in time $O(n^{c+1})$.
For $C, C' \in \mathcal{K}$, one can decide whether $H$ has a $C$--$C'$ path in time $O(n^{c})$.

Let $C := S \cap K$ and $C' := S' \cap K$.
Now, by Corollary~\ref{cor:split_maximal}, it suffices to show that 
$T_{C} \sevstep T_{C'} $ if and only if there is a path between $C$ and $C'$ in $H$.

Assume that $T_{C} \sevstep T_{C'}$.
Let $\langle S_{1} = T_{C}, S_{2}, \dots, S_{p} = T_{C'} \rangle$ be a reconfiguration sequence from $T_{C}$ to $T_{C'}$,
and let $C_{i} = S_{i} \cap K$ for $1 \le i \le p$. Observe that $|\symdiff{C_{i}}{C_{i+1}}| \le 1$ for each $1 \le i < p$.
If $C_{i} \ne C_{i+1}$, Corollary~\ref{cor:split_maximal} and Lemma~\ref{lem:split-edge}
imply that $\{C_{i}, C_{i+1}\} \in \mathcal{E}$.
Since $C_{1} = C$ and $C_{p} = C'$, we can conclude that $H$ has a $C$--$C'$ path.

Next assume that there is a path between $C$ and $C'$ in $H$.
Let $(C_{1} = C, C_{2}, \dots, C_{q} = C')$ be such a path.
Lemma~\ref{lem:split-edge} and the definition of $H$ together imply that
$T_{C_{i}} \sevstep T_{C_{i+1}}$ for each $1 \le i < q$.
Since $C_{1} = C$ and $C_{q} = C'$, we have $T_{C} \sevstep T_{C'}$.
\end{proof}


\subsection{PSPACE-completeness when $c$ is a part of input}

For split graphs with $c$ as a part of input, NP-completeness of \textsc{Colorable Set}
is shown in \cite{YannakakisG87,CorneilF89} by a reduction from \textsc{Set Cover}.
Here we present a reduction essentially the same as theirs but from \textsc{Independent Set Reconfiguration} under $\TJ$,
which is equivalent to $\CSR_{\TJ}$ with $c = 1$.

\begin{theorem}
\label{thm:split-hard}
Given a split graph and $c$-colorable sets $S$ and $S'$ of size $k$ in the graph,
it is PSPACE-complete to decide whether $S \sevstep S'$
under any of $\TS$, $\TJ$, and $\TAR{k-1}$.
\end{theorem}
\begin{proof}
The problem is clearly in PSPACE under all rules.
To show the PSPACE-hardness, we present a polynomial-time 
reduction from \textsc{Independent Set Reconfiguration} under $\TJ$,
which is PSPACE-complete~\cite{KaminskiMM12}.

Let $G$ be a graph on which \textsc{Independent Set Reconfiguration} is considered.
From $G$, we construct a split graph $H = (V(G), E(G); E)$ as follows:
in $H$, $V(G)$ is a clique and $E(G)$ is an independent set,
and $\{v, e\} \in E$ if and only if $v$ is not an endpoint of $e$ for $v \in V(G)$ and $e \in E(G)$.
We define a map $\phi$ from independent sets $I$ of size $|V(G)| - c$ in $G$ 
to $c$-colorable sets of size $|E(G)| + c$ in $H$
as $\phi(I) = V(H) \setminus I$.

\begin{claim}
  \label{clm:split-hard-colorable}
  For every independent set $I$ of size $|V(G)| - c$ in $G$, $\phi(I)$ is $c$-colorable in $H$.
\end{claim}
\begin{proof}
[Proof of Claim~\ref{clm:split-hard-colorable}]
Let $C = V(G) \setminus I$.
Since $C$ is a vertex cover of size $c$ in $G$,
it holds that $C \not\subseteq N_{H}(e)$ for each $e \in E(G)$.
Hence, the split graph $H[\phi(I)] = H[C \cup E(G)]$ has no clique of size larger than $|C| = c$.
\end{proof}

The following claim implies that $\phi$ is a bijection.
\begin{claim}
  \label{clm:split-reduction-bijection}
  If $S$ is a $c$-colorable set of size $|E(G)| + c$ in $H$,
  then $E(G) \subseteq S$
  and $V(G) \setminus S$ is an independent set of size $|V(G)| - c$ in $G$.
 \end{claim}
\begin{proof}
[Proof of Claim~\ref{clm:split-reduction-bijection}]
Since $S$ is $c$-colorable and $V(G)$ is a clique in $H$,
we have $|S \cap V(G)| \le c$.
This implies that $|S \cap E(G)| \ge |S| - |S \cap V(G)| \ge |E(G)|$,
and thus $E(G) \subseteq S$.

Let $C := S \cap V(G)$.
The discussion above implies that $|C| = c$.
Since $S$ is $c$-colorable, each vertex in $S \cap E(G)$ has a non-neighbor in $C$.
Therefore, $C$ is a vertex cover of size $c$ in $G$.
In other words, $V(G) \setminus C = V(G) \setminus S$ is an independent set of size $|V(G)| - c$ in $G$.
\end{proof}

Let $I$ and $I'$ be independent sets of size $|V(G)| - c$ in $G$.
From the discussion above, it suffices to show that
the following properties are equivalent:
\begin{enumerate}
  \item $I \onestep I'$ under $\TJ$; \label{itm:split-onestep-is}
  \item $\phi(I) \onestep \phi(I')$ under $\TS$; \label{itm:split-onstep-TS}
  \item $\phi(I) \onestep \phi(I')$ under $\TJ$; \label{itm:split-onstep-TJ}
  \item $\phi(I) \onestep \phi(I) \cap \phi(I') \onestep \phi(I')$ under $\TAR{|E(G)| + c - 1}$. \label{itm:split-onstep-TAR}
\end{enumerate}

To show that Property~\ref{itm:split-onestep-is} implies Property~\ref{itm:split-onstep-TS},
assume that $I \onestep I'$ under $\TJ$.
Let $v \in I \setminus I'$ and $v' \in I' \setminus I$.
Now $\phi(I) \setminus \phi(I') = v'$ and $\phi(I') \setminus \phi(I) = v$ hold.
Since $v$ and $v'$ are in the same clique $V(G)$ of $H$, we have $\phi(I) \onestep \phi(I')$ under $\TS$.

By the definitions of $\TS$ and $\TJ$,
Property~\ref{itm:split-onstep-TS} implies Property~\ref{itm:split-onstep-TJ}.
By Lemma~\ref{lem:TJ=TAR}, Property~\ref{itm:split-onstep-TJ} and Property~\ref{itm:split-onstep-TAR} are equivalent.

We finally show that Property~\ref{itm:split-onstep-TJ} implies Property~\ref{itm:split-onestep-is}.
Assume that $\phi(I) \onestep \phi(I')$ under $\TJ$.
This implies that $|\phi(I) \setminus \phi(I')| = |\phi(I') \setminus \phi(I)| = 1$.
By Claim~\ref{clm:split-reduction-bijection}, $\symdiff{\phi(I)}{\phi(I')} \subseteq V(G)$.
Therefore, $|I \setminus I'| = |I' \setminus I| = 1$ also holds.
Thus $I \onestep I'$ under $\TJ$.
\end{proof}


\section{Co-comparability graphs}
\label{sec:co-comp}

In \textsc{Shortest Path Reconfiguration}, we are given an unweighted graph $G = (V,E)$, two vertices $s,t \in V$,
and two shortest $s$--$t$ paths $P$ and $P'$ in $G$.
The goal is to decide whether there is a sequence of shortest $s$--$t$ paths $\langle P_{1} =P, P_{2}, \dots, P_{p} = P' \rangle$
such that for each $1 \le i < p$, 
$P_{i}$ and $P_{i+1}$ differ at exactly one vertex.

To show the PSPACE-hardness,
we generalize the reduction for \textsc{Independent Set Reconfiguration} in \cite{KaminskiMM12}.
Note that the following theorem holds even if $c$ is a fixed constant.
\begin{theorem}
\label{thm:cocomp-hard}
Given a co-comparability graph and $c$-colorable sets $S$ and $S'$ of size $k$ in the graph,
it is PSPACE-complete to decide whether $S \sevstep S'$
under any of $\TS$, $\TJ$, and $\TAR{k-1}$.
\end{theorem}
\begin{proof}
Let $G = (V,E)$ be a graph, $P$ and $P'$ be two shortest $s$--$t$ paths  in $G$ for some $s, t \in V$.
Let $\ell$ be the distance from $s$ to $t$ in $G$.
We denote by $D_{i} \subseteq V$ the vertices of distance $i$ from $s$ and distance $\ell -i$ from $t$.
Observe that for every shortest $s$--$t$ path $Q$ and for every $i \in \{0, \dots, \ell\}$, it holds that $|V(Q) \cap D_{i}| = 1$.

Now we define a graph $H$. The vertex set $V(H)$ is $\bigcup_{0 \le i \le \ell} D_{i}$.
Each $D_{i}$ is a clique in $H$.
Two vertices $u \in D_{i}$ and $v \in D_{i+1}$ are adjacent in $H$ if they are not adjacent in $G$.
There is no other edge in $H$.
For $c = 1$, the reduction is already done~\cite{KaminskiMM12}.
If $c \ge 2$, then we modify the graph $H$ and obtain $H'$.
For each $i \in \{0,\dots,\ell\}$, let $C_{i}$ be a clique of $c-1$ new vertices.
The graph $H'$ has the vertex set $V(H') = V(H) \cup \bigcup_{0 \le i \le \ell} C_{i}$,
and it has all edges in $H$. Additionally, all vertices in $C_{i}$ are adjacent to all vertices in $D_{i} \cup D_{i+1}$
(assuming that $D_{\ell+1} = \emptyset$).

\begin{claim}
  \label{clm:cocomp-H'}
  $H'$ is a co-comparability graph.
\end{claim}
\begin{proof}
[Proof of Claim~\ref{clm:cocomp-H'}]
Let $\prec'$ be the partial ordering on $V(H')$ such that
$u \prec' v$ if and only if $u \in C_{i} \cup D_{i}$ and $v \in C_{j} \cup D_{j}$ for some $i < j$.
Let $\prec$ be a linear extension of $\prec'$.

Assume that $u \prec v \prec w$ and $\{u,w\} \in E$.
Let $u \in C_{i} \cup D_{i}$.
From the construction of $H'$, the vertices $v$ and $w$ are in $(C_{i} \cup D_{i}) \cup (C_{i+1} \cup D_{i+1})$.
If $v \in C_{i} \cup D_{i}$, then $\{u,v\} \in E$.
Otherwise, we have $v, w \in C_{i+1} \cup D_{i+1}$, and thus $\{v,w\} \in E$.
\end{proof}

Now we define a map $\phi$ from shortest $s$--$t$ paths in $G$ to $c$-colorable sets of size $(\ell + 1) c$ in $H'$.
For a shortest $s$--$t$ path $P = (v_{0}, \dots, v_{\ell})$ in $G$, we set $\phi(P) = V(P) \cup \bigcup_{0 \le i \le \ell} C_{i}$.
In $H'[\phi(P)]$, a maximal clique is formed either by $C_{i}$ and $v_{i}$ or by $C_{i}$ and $v_{i+1}$ for some $i$.
(Recall that $v_{i}$ and $v_{i+1}$ are not adjacent in $H'$ for all $i$.)
Thus, $\phi(P)$ is $c$-colorable in $H'$.
The following claim implies that $\phi$ is a bijection.
\begin{claim}
  \label{clm:cocomp-reduction-bijection}
  If $S$ is a $c$-colorable set of size $(\ell + 1) c$ in $H'$,
  then $\bigcup_{0 \le i \le \ell} C_{i} \subseteq S$
  and $S \setminus \left(\bigcup_{0 \le i \le \ell} C_{i}\right)$ induces a shortest $s$--$t$ path in $G$.
 \end{claim}
\begin{proof}
[Proof of Claim~\ref{clm:cocomp-reduction-bijection}]
Since $|S| = (\ell + 1) c$ and $C_{i} \cup D_{i}$ is a clique for each $i \in \{0,\dots,\ell\}$,
we have $|S \cap (C_{i} \cup D_{i})| = c$.
It holds that $C_{0} \subseteq S$ as $|C_{0} \cup D_{0}| = c$.
Assume that $C_{i} \subseteq S$ for some $i \ge 0$.
Since all $c-1$ vertices in $C_{i}$ are adjacent to all vertices in $D_{i+1}$,
the set $S$ can contain at most one vertex in $D_{i+1}$.
Thus, $|S \cap (C_{i+1} \cup D_{i+1})| = c$ implies that $C_{i+1} \subseteq S$.

From the discussion above, we can also see that $|S \cap D_{i}| = 1$ for all $i \in \{0,\dots,\ell\}$.
For each $i$, let $v_{i}$ be the unique vertex in $S \cap D_{i}$.
Note that $v_{0} = s$ and $v_{\ell} = t$.
Observe that $v_{i}$ and $v_{i+1}$ are not adjacent in $H'$ for each $i$,
since otherwise $C_{i} + v_{i} + v_{i+1}$ becomes a clique of size $c+1$.
Therefore, $(v_{0}, \dots, v_{\ell})$ is an $s$--$t$ path in $G$.
Since its length is $\ell$, it is shortest.
\end{proof}

Let $P$ and $P'$ be shortest $s$--$t$ path in $G$.
To complete the proof of Theorem~\ref{thm:cocomp-hard}, it suffices to show that the following properties are equivalent:
\begin{enumerate}
  \item $P$ and $P'$ differ at exactly one vertex; \label{itm:cocomp-onestep-path}
  \item $\phi(P) \onestep \phi(P')$ under $\TS$; \label{itm:cocomp-onstep-TS}
  \item $\phi(P) \onestep \phi(P')$ under $\TJ$; and \label{itm:cocomp-onstep-TJ}
  \item $\phi(P) \onestep \phi(P) \cap \phi(P') \onestep \phi(P')$ under $\TAR{(\ell + 1)c - 1}$. \label{itm:cocomp-onstep-TAR}
\end{enumerate}

To show that Property~\ref{itm:cocomp-onestep-path} implies Property~\ref{itm:cocomp-onstep-TS},
assume that $P = (v_{0}, \dots, v_{\ell})$ and $P' = (v_{0}', \dots, v_{\ell}')$ 
differ at exactly one index $i \in \{1, \dots, \ell-1\}$.
Hence $\phi(P) \setminus \phi(P') = v_{i}$ and $\phi(P') \setminus \phi(P) = v_{i}'$.
Since $v_{i}$ and $v_{i}'$ are in the clique $D_{i}$, we have $\phi(P) \onestep \phi(P')$ under $\TS$.

By the definitions of $\TS$ and $\TJ$,
Property~\ref{itm:cocomp-onstep-TS} implies Property~\ref{itm:cocomp-onstep-TJ}.
By Lemma~\ref{lem:TJ=TAR}, Property~\ref{itm:cocomp-onstep-TJ} and Property~\ref{itm:cocomp-onstep-TAR} are equivalent.

Finally, we show that Property~\ref{itm:cocomp-onstep-TJ} implies Property~\ref{itm:cocomp-onestep-path}.
Assume that $\phi(P) \onestep \phi(P')$ under $\TJ$.
This implies that $|\phi(P) \setminus \phi(P')| = |\phi(P') \setminus \phi(P)| = 1$.
Furthermore, Claim~\ref{clm:cocomp-reduction-bijection} implies that 
$\phi(P) \setminus \phi(P') = V(P) \setminus V(P')$ and
$\phi(P') \setminus \phi(P) = V(P') \setminus V(P)$.
Therefore, $P$ and $P'$ differ exactly at one vertex.
\end{proof}

\section{Concluding remarks}
\label{sec:conclusion}
We show that \textsc{Colorable Set Reconfiguration} under $\TARrule$/$\TJ$ is linear-time solvable on interval graphs.
Our results give a sharp contrast of the computational complexity with respect to graph classes, while some cases are left unanswered.
One of the main unsettled cases is $\CSR_{\TARrule}$ with fixed $c > 1$ for chordal graphs (see Table~\ref{tbl:summary}).
In particular, what is the complexity of $\CSR_{\TARrule}$ with $c=2$ for chordal graphs?
This problem is equivalent to the reconfiguration of feedback vertex sets under $\TARrule$ on chordal graphs.
It would be also interesting to study the shortest variant on split graphs with a constant $c$.

Our positive results for $\CSR_{\TARrule}$ on interval graphs and split graphs (Theorems~\ref{thm:interval} and \ref{thm:split_fixed}) 
do not imply analogous results for $\CSR_{\TS}$.
The complexity of $\CSR_{\TS}$ is not settled for these graph classes even with a fixed constant $c$.
It was only recently shown that if $c = 1$, then $\CSR_{\TS}$ can be solved in polynomial time for interval graphs~\cite{BonamyB17}.
For $c \ge 2$, $\CSR_{\TS}$ on interval graphs is left unsettled.
For split graphs, although co-NP-hardness of a related problem is known~\cite{BonamyB17},
$\CSR_{\TS}$ is not solved for all $c \ge 1$.





\bibliographystyle{plainurl}
\bibliography{csr}


\end{document}